\documentclass[11pt]{article}
\usepackage[latin1]{inputenc}
\usepackage{fullpage}  
\usepackage{hyperref}

\usepackage[compatibility=false]{caption}
\usepackage{subcaption}
\usepackage{tikz}

\usepackage{amsmath}  
\usepackage{amsfonts}  
\usepackage{amssymb}  
\usepackage{amsthm}    

\usepackage{cleveref}
\crefname{lemma}{lemma}{Lemmas}

\usepackage{mleftright}

\newtheorem{theorem}{Theorem}
\newtheorem{proposition}{Proposition}
\newtheorem{corollary}{Corollary}
\newtheorem{lemma}{Lemma}

\newtheorem{definition}{Definition}
\newtheorem{conjecture}{Conjecture}

\newcommand{\FO}{\textsf{FO}}
\newcommand{\AC}{\textsf{AC}}
\newcommand{\NC}{\textsf{NC}}
\renewcommand{\L}{\textsf{L}}
\newcommand{\SL}{\textsf{SL}}
\newcommand{\NL}{\textsf{NL}}
\renewcommand{\P}{\textsf{P}}
\newcommand{\NP}{\textsf{NP}}
\newcommand{\MPC}{\textsf{MPC}}
\newcommand{\DMPC}{\textsf{DMPC}}

\title{Equivalence Classes and Conditional Hardness in Massively Parallel Computations}
\author{Danupon Nanongkai\\{\small KTH Royal Institute of Technology}\\{\small {\tt danupon@gmail.com}}
\and Michele Scquizzato\\{\small University of Padova}\\{\small {\tt scquizza@math.unipd.it}}}
\date{}

\begin{document}

\maketitle
\thispagestyle{empty}

\begin{abstract}
	The \emph{Massively Parallel Computation} (MPC) model serves as a common abstraction of many modern
	large-scale data processing frameworks, and has been receiving increasingly more attention over the past few years,
	especially in the context of classical graph problems. So far, the only way to
	argue lower bounds for this model is to condition on conjectures about the
	hardness of some specific problems, such as graph connectivity on promise graphs that are
	either one cycle or two cycles, usually called the \emph{one cycle vs.\ two cycles} problem.
	This is unlike the traditional arguments based on
	conjectures about complexity classes (e.g., $\P \neq \NP$), which are often
	more robust in the sense that refuting them would lead to groundbreaking
	algorithms for a whole bunch of problems.
	
	In this paper we present connections between problems and classes of
	problems that allow the latter type of arguments. These connections concern
	the class of problems solvable in a sublogarithmic amount of rounds in the MPC model,
	denoted by $\MPC(o(\log N))$, and some
	standard classes concerning space complexity, namely $\L$ and $\NL$, and suggest
	conjectures that are robust in the sense that refuting them would lead to many
	surprisingly fast new algorithms in the MPC model. We also obtain new conditional lower bounds,
	and prove new
	reductions and equivalences between problems in the MPC model. Specifically, our main
	results are as follows.
	
	\begin{itemize}
		\item Lower bounds conditioned on the one cycle vs.\ two cycles conjecture
		can be instead argued under the $\L \nsubseteq \MPC(o(\log N))$ conjecture:
		these two assumptions are equivalent, and refuting
		either of them would lead to $o(\log N)$-round MPC algorithms for a large
		number of challenging problems, including list ranking, minimum cut, and
		planarity testing. In fact, we show that these problems and many others
		require asymptotically the same number of rounds as the seemingly much
		easier problem of distinguishing between a graph being one cycle or two cycles.
		
		\item Many lower bounds that were argued under the one cycle vs.\ two cycles 
		conjecture can be argued under an even more robust (thus harder
		to refute) conjecture, namely $\NL \nsubseteq \MPC(o(\log N))$. Refuting this
		conjecture would lead to $o(\log N)$-round MPC algorithms for an even
		larger set of problems, including all-pairs shortest paths,
		betweenness centrality, and all aforementioned ones.
		Lower bounds under this conjecture hold for problems such as
		perfect matching and network flow.
	\end{itemize}
	
\end{abstract}

\newpage
\setcounter{page}{1}

\section{Introduction}\label{sec:intro}

The \emph{Massively Parallel Computation} (MPC) model is arguably the most popular model of computation that
captures the essence of several very successful general-purpose frameworks for massively parallel coarse-grained
computations on large data sets, such as MapReduce~\cite{DeanG08}, Hadoop~\cite{White12}, Spark~\cite{ZahariaXWDADMRV16},
and  Dryad~\cite{IsardBYBF07}. The main feature of this model is that a single commodity machine of a large cluster cannot store
the entirety of the input, but just a sublinear fraction of it. This is an important restriction since we think of the data set as
being very large. The computation proceeds in synchronous rounds, and in each of them the machines can exchange data
with each other with the sole restriction that no one can send or receive more data than it is capable of storing. The goal
is to keep the total number of rounds as low as possible.

This basic model has been much investigated in the past decade, mostly from an algorithmic point of view~\cite{KarloffSV10,SuriV11,LattanziMSV11,GoodrichSZ11,EneIM11,BahmaniMVKV12,BahmaniKV12,PietracaprinaPRSU12,RastogiMCS13,AfratiSSU13,KumarMVV15,HegemanP15,AndoniNOY14,JacobLS14,KiverisLMRV14,FishKLRT15,BeameKS17,ImMS17,BateniBDHKLM17,RoughgardenVW18,BoroujeniEGHS18,AhnG18,BehnezhadDH18,CzumajLMMOS18,YaroslavtsevV18,HarveyLL18,GhaffariGKMR18,AndoniSSWZ18,AssadiSW19,AssadiCK19,AssadiBBMS19,GhaffariU19,HajiaghayiSS19,LackiMW18,GamlathKMS19,ImM19,ItalianoLMP19,BehnezhadDELSM19,AndoniSZ19,BehnezhadBDFHKU19,ChangFGUZ19,BehnezhadHH19,GhaffariKU19,BehnezhadDELM19}. 
A common outcome is that, when $N$ denotes the input size, a solution terminating in $O(\log N)$ rounds is possible,
usually by simulating known PRAM algorithms~\cite{KarloffSV10,GoodrichSZ11}, but going below that resisted the
efforts of many researchers. Recently, a few works managed to break the $O(\log N)$ barrier by relaxing a bit the
sublinear constraint on the memory size, and showed that some graph problems allow for $o(\log N)$-round
solutions in the so-called \emph{near-linear memory} regime, whereby machines have memories of size $\tilde O(n)$,
where $n$ is the number of nodes in the graph~\cite{CzumajLMMOS18,GhaffariGKMR18,AssadiBBMS19,AssadiCK19,BehnezhadHH19}.\footnote{Notice that this relaxes the sublinear constraint on the memory size in the case of sparse graphs.} However,
without this kind of relaxations only a handful of problems are known to admit a $o(\log N)$-round algorithm~\cite{GhaffariU19,HajiaghayiSS19,ChangFGUZ19}.\footnote{Some algorithms have been analyzed in terms of other
parameters, such as the diameter~\cite{AndoniSSWZ18,AndoniSZ19,BehnezhadDELM19} or the spectral gap~\cite{AssadiSW19}
of the graph. The round complexity of these algorithms is $o(\log N)$ in some cases, but it remains $\Omega(\log N)$
in general. In this paper we do not consider this kind of parameterized analysis.} A fundamental question is thus whether
many known $O(\log N)$-round algorithms can be complemented with tight lower bounds. 

Unfortunately, proving {\em unconditional} lower bounds---that is, without any assumptions---seems extremely
difficult in this model, as it would imply a breakthrough in circuit complexity: Roughgarden~et~al.~\cite{RoughgardenVW18}
showed that, when enough machines are available, proving any super-constant lower bound for any problem in $\P$ would
imply new circuit lower bounds, and specifically would separate $\NC^1$ from $\P$---a long-standing open question in
complexity theory that is a whisker away from the $\P$ vs.\ $\NP$ question. This means that the lack of super-constant
lower bounds in the MPC model can be blamed on our inability to prove some computational hardness results.

In light of this barrier, the focus rapidly shifted to proving \emph{conditional} lower bounds, that is, lower bounds conditioned
on plausible hardness assumptions. One widely-believed assumption concerns graph connectivity, which, when machines
have a memory of size $O(n^{1-\epsilon})$ for a constant $\epsilon > 0$, is conjectured to require $\Omega(\log n)$
MPC rounds~\cite{KarloffSV10,RastogiMCS13,BeameKS17,RoughgardenVW18,YaroslavtsevV18}.\footnote{Observe that in the
near-linear memory regime this conjecture breaks: graph connectivity can be solved in $O(1)$ MPC rounds~\cite{BehnezhadDH18}.}
The same conjecture is often made even for the special case of the problem where the graph consists of either one cycle or two
cycles, usually called {\em one cycle vs.\ two cycles} problem. The one cycle vs.\ two cycles conjecture has been proven useful to
show conditional lower bounds for several problems, such as maximal independent set, maximal matching~\cite{GhaffariKU19},
minimum spanning trees in low-dimensional spaces~\cite{AndoniNOY14}, single-linkage clustering~\cite{YaroslavtsevV18},
$2$-vertex connectivity~\cite{AndoniSZ19}, generation of random walks~\cite{LackiMOS}, as well as parameterized conditional
lower bounds~\cite{BehnezhadDELM19}.\footnote{The one cycle vs.\ two cycles problem is usually stated such that, in the
case of two cycles, these have $n/2$ nodes each. However, we observe that all the mentioned conditional lower bounds hold
also when the two cycles may have arbitrary lengths.}

However, it is not clear whether the one cycle vs.\ two cycles conjecture is true or not, and if not, what its refutation implies.
This situation is in contrast with traditional complexity theory, where a refutation of a conjectured relationship between
complexity classes would typically imply groundbreaking algorithmic results for a large number of problems; for example,
if the $\P \neq \NP$ conjecture fails, then there would be efficient (polynomial-time) algorithms for {\em all} problems
in $\NP$, including a number of ``hard'' problems. To put it another way, a conjecture like $\P \neq \NP$ is more
{\em robust} in the sense that it is extremely hard to refute---doing so requires a major algorithmic breakthrough. 
The goal of this paper is to explore conjectures of this nature in the MPC model.

\subsection{Summary of Contributions}

In this paper we show many connections between problems and classes of problems that lead to more robust conjectures
for the MPC model. In particular, we study the connections between the class of problems solvable in a sublogarithmic
amount of rounds in the MPC model with $O(N^{1-\epsilon})$ memory per machine for some constant $\epsilon \in (0,1)$
and up to polynomially many machines, denoted by $\MPC(o(\log N))$, and the standard space complexity classes $\L$ and
$\NL$. (Recall that $\L$ and $\NL$ are the classes of decision problems decidable in logarithmic space on deterministic
and nondeterministic Turing machines, respectively.) The connection between MPC and these complexity classes is
enabled by a recent result showing how Boolean circuits can be efficiently simulated in the MPC model. In short, we present
a set of observations and reductions that suggest that $\L \nsubseteq \MPC(o(\log N))$ and $\NL \nsubseteq \MPC(o(\log N))$
are two robust conjectures that might play crucial roles in arguing lower bounds in the MPC model, as they already imply tight
conditional lower bounds for a large number of problems. In particular, with some assumptions on the total amount of
memory (equivalently, machines) available in the system, we can conclude the following. 

\begin{enumerate}
	\item\label{item:one} {\bf Robustness:} The one cycle vs.\ two cycles conjecture is robust, since it is equivalent to conjecturing that $\L \nsubseteq \MPC(o(\log N))$, and refuting this conjecture requires showing $o(\log N)$-round algorithms for all problems in $\L$. This class includes many important problems such as graph connectivity, cycle detection, and planarity testing.
	
\begin{figure}[h]
    \begin{center}
       \includegraphics[trim=0cm 0.1cm 0cm 0.1cm, clip=true,scale=0.69]{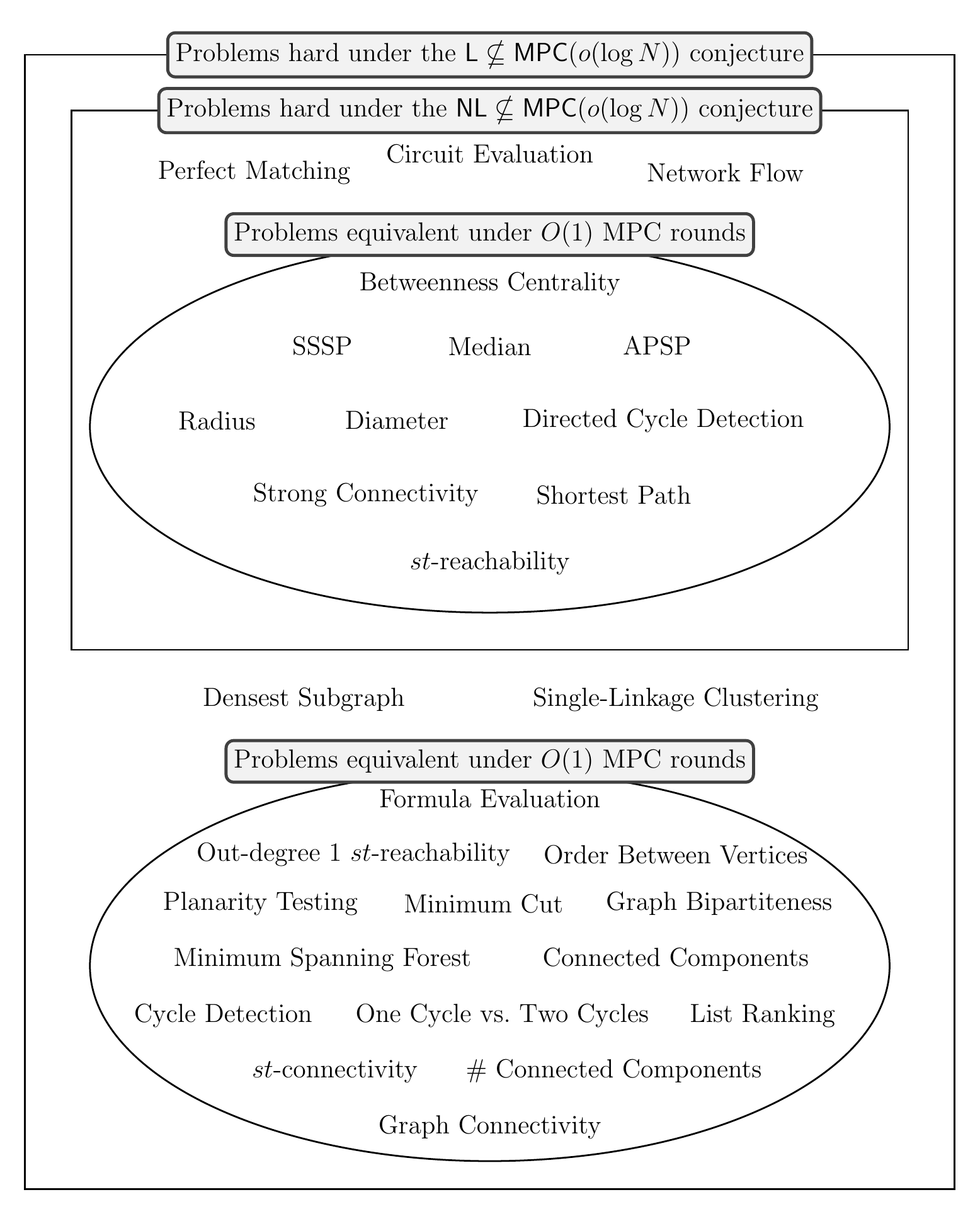}
       \caption{A classification of the complexity of some prominent problems in the MPC model.
       Problems in the top ellipse are on unweighted graphs.}\label{fig:results}
    \end{center}
\end{figure}

	\item {\bf Equivalences:} All $\L$-complete problems are {\em equivalent} in the sense that they require asymptotically the same number of rounds.  This means that the one cycle vs.\ two cycles problem, which is $\L$-complete (see Appendix~\ref{sec:space_1vs2cycles}), is equivalent to many seemingly harder problems, such as graph bipartiteness, minimum cut, and formula evaluation (see problems in the bottom ellipse in Figure~\ref{fig:results} for more). This also means that the conjectures on the hardness of graph connectivity and on the hardness of the one cycle vs.\ two cycles problem are equivalent. 

	 Additionally, all $\NL$-complete problems and a few others are also equivalent. These problems include $st$-reachability, all-pairs shortest paths (both the directed and undirected cases) on unweighted graphs, diameter, and betweenness centrality (see problems in the top ellipse in Figure~\ref{fig:results} for more).

	\item {\bf New conditional lower bounds:} Assuming the one cycle vs.\ two cycles conjecture (equivalently, $\L \nsubseteq \MPC(o(\log N))$), there are no $o(\log N)$-round algorithms for all $\L$-hard problems and a few other problems. This implies new conditional lower bounds for more than a dozen of problems, such as betweenness centrality, planarity testing, graph bipartiteness, list ranking, formula evaluation, and densest subgraph  (see problems in the big rectangle in Figure~\ref{fig:results} for more). Previously only a few lower bounds were known, e.g., for single-linkage clustering~\cite{YaroslavtsevV18} and maximum matching~\cite{Onak19}. (Of course, lower bounds for connectivity-related problems are trivially implied by the  one cycle vs.\ two cycles conjecture.) Most of our lower bounds are tight (e.g., lower bounds for problems in the ellipses in Figure~\ref{fig:results}).

	\item \label{item:four} {\bf A more robust conjecture.} For $\NL$-hard problems, we can argue lower bounds under the more robust $\NL \nsubseteq \MPC(o(\log N))$ conjecture. These problems include perfect matching, single-source shortest paths, diameter, and network flow (see problems in the small rectangle in Figure~\ref{fig:results} for more). Note that, since $\L \subseteq \NL$, the $\NL \nsubseteq \MPC(o(\log N))$ conjecture is more robust (i.e., safer, more likely to be true) than its counterpart with $\L$. 
\end{enumerate}

\subsection{Related Work}
Fish et al.~\cite{FishKLRT15} were perhaps the first to establish a connection between the MPC model
and classical complexity classes. Besides the introduction of a uniform version of the model, they showed
that constant-round MPC computations can simulate sublogarithmic space-bounded Turing machines, and
then proved strict hierarchy theorems for the MPC model under certain complexity-theoretic assumptions.

Roughgarden et al.~\cite{RoughgardenVW18} discuss connections between the MPC model and Boolean
circuits. They show that standard degree arguments for circuits can be applied to MPC computations
as well, and specifically that any Boolean function whose polynomial representation has degree $d$ requires
$\Omega(\log_s d)$ rounds of MPC using machines with memory $s$. This implies an $\Omega(\log_s n)$
lower bound on the number of rounds for graph connectivity. Perhaps more interestingly, the authors show
a barrier for unconditional lower bounds by observing that, if enough machines are available, then proving
any super-constant lower bound in the MPC model for any problem in $\P$ would imply new circuit lower bounds,
and specifically would separate $\NC^1$ from $\P$, thus answering a notorious open question in circuit complexity.
This result follows by showing that, with a number of available machines polynomial in the number of input nodes
of the circuit, $\NC^1$ circuits can be efficiently simulated in the MPC model. We observe that their argument
readily generalizes to show that any bounded fan-in Boolean circuit of depth $d$ and of polynomial size can be
simulated in $O(\left\lceil d/\log s \right\rceil)$ MPC rounds. Very recently, Frei and Wada~\cite{FreiW19} prove
the same result improving over the amount of machines required for the simulation---from linear to strongly
sublinear in the size of the circuit.

Given the difficulty of proving lower bounds for all algorithms, one can (a) prove lower bounds for restricted classes
of algorithms, or (b) prove conditional lower bounds: assume one lower bound, and transfer the conjectured
hardness to other problems via reductions (with common examples being the theory of $\NP$-hardness and its
more recent analogue for problems in $\P$, usually called \emph{fine-grained complexity theory}). Both paths
give a deep understanding and warn us what not to try when designing algorithms.

Within the first line of inquiry, Pietracaprina et al.~\cite{PietracaprinaPRSU12} prove lower bounds for matrix
multiplication algorithms that compute all the $n^3$ elementary products (thus ruling out Strassen-like
algorithms). Similar kinds of limitations are required by Beame et al.~\cite{BeameKS17}, Jacob et al.~\cite{JacobLS14},
Im and Moseley~\cite{ImM19}, and Assadi and Khanna~\cite{AssadiK18} to prove lower bounds for $st$-connectivity,
list ranking, graph connectivity, and maximum coverage, respectively. Of a similar flavor are the results of Afrati
et al.~\cite{AfratiSSU13}, who show, for a fixed number of rounds (usually a single round), space-communication
tradeoffs.

Within the second line of inquiry fall~\cite{AndoniNOY14,YaroslavtsevV18,AndoniSZ19,LackiMOS}, which use the
conjecture on the hardness of graph connectivity as a hardness assumption for proving conditional lower bounds
for other problems such as minimum spanning trees in low-dimensional spaces, single-linkage clustering,
$2$-vertex connectivity, and generating random walks, respectively. Very recently, Ghaffari et al.~\cite{GhaffariKU19}
present conditional lower bounds for other key graph problems such as constant-approximate maximum matching,
constant-approximate vertex cover, maximal independent set, and maximal matching. Their lower bounds also rest on
the hardness of graph connectivity, and are obtained by introducing a new general method that lifts (unconditional) lower
bounds from the classical LOCAL model of distributed computing to the MPC model. Assuming the same conjecture,
Behnezhad et al.~\cite{BehnezhadDELM19} show a parameterized lower bound of $\Omega(\log D)$ for identifying
connected components in graphs of diameter $D$. By observing that a couple of specific $\NC^1$ reductions can be
simulated in $O(1)$ MPC rounds, Dhulipala et al.~\cite{DhulipalaDKPSS20} show that if a variant of graph connectivity on
batch-dynamic graphs can be solved within a certain amount of rounds, so can all the problems in $\P$. A conditional
lower bound following a different kind of argument is given by Andoni et al.~\cite{AndoniSSWZ18}, who show that an
$n^{o(1)}$-round MPC algorithm that answers $O(n+m)$ pairs of reachability queries in directed graphs with $n$
nodes and $m$ edges can be simulated in the RAM model yielding faster Boolean matrix multiplication algorithms.

Several other models have been developed in the quest to establish rigorous theoretical foundations of (massively)
parallel computing, with the PRAM being one of the most investigated. The MPC model is more powerful than the PRAM
since PRAM algorithms can be simulated in the MPC model with constant slowdown~\cite{KarloffSV10,GoodrichSZ11},
and some problems (such as evaluating the XOR function) can be solved much faster in the MPC model.

Valiant's \emph{bulk-synchronous parallel} (BSP) model~\cite{Valiant90} anticipated many of the features of MPC-type
computations, such as the organization of the computation in a sequence of synchronous rounds (originally called
\emph{supersteps}). Several papers (e.g., \cite{Goodrich99,MacKenzieR98,AdlerDJKR98,ScquizzatoS14,BilardiSS18})
explored the power of this model by establishing lower bounds on the number of supersteps or on the communication
complexity required by BSP computations, where the latter is defined as the sum, over all the supersteps of an algorithm,
of the maximum number of messages sent or received by any processor. Lower bounds on the number of supersteps are
usually of the form $\Omega(\log_h N)$, where $h$ is the maximum number of messages sent or received by any processor
in any superstep.

Another model aiming at serving as an abstraction for modern large-scale data processing frameworks is the
\emph{$k$-machine} model~\cite{KlauckNPR15}. Partly inspired by message-passing models in distributed computing,
in the $k$-machine model there are $k$ available machines, and in each round any pair of machines is allowed to communicate
using messages of a given size. Hard bounds on the point-to-point communication lead to very strong round lower
bounds in this model~\cite{KlauckNPR15,PanduranganRS18,PanduranganRS18a}.

The \emph{congested clique} (see, e.g., \cite{DruckerKO13}) is a model for network computations bearing
some similarities with the MPC model. On one hand, algorithms for this model can be simulated in the MPC
model---under some specific conditions on the size of the local memories~\cite{HegemanP15,GhaffariGKMR18,BehnezhadDH18}.
On the other hand, analogously to the MPC model, proving a super-constant unconditional lower bound in the
congested clique for a problem in $\NP$ would imply better circuit size-depth tradeoffs for such a problem than
are currently known~\cite{DruckerKO13}. This induced further investigations of the model under the lens of complexity
theory~\cite{KorhonenS18}.


\section{Preliminaries}

\subsection{The MPC Model}
The \emph{Massively Parallel Computation (MPC)} model is a theoretical abstraction capturing the main distinguishing
aspects of several popular frameworks for the parallel processing of large-scale datasets. It was introduced by Karloff,
Suri, and Vassilvitskii~\cite{KarloffSV10}, and refined in subsequent work~\cite{GoodrichSZ11,BeameKS17,AndoniNOY14}.

In this model the system consists of $p$ identical machines (processors), each with a local memory of size $s$. If $N$
denotes the size of the input, then $s = O(N^{1-\epsilon})$ for some constant $\epsilon > 0$, and the total amount
of memory available in the system is $p \cdot s = O(N^{1+\gamma})$ for some constant $\gamma \geq 0$. The space
size is measured by words, each of $\Theta(\log N)$ bits. Initially, the input is adversarially distributed across the machines.
The computation proceeds in synchronous rounds. In each round, each machine performs some computation on the
data that resides in its local memory, and then, at the end of the round, exchanges messages with other machines.
The total size of messages sent or received by each machine in each round is bounded by $s$.\footnote{This means
that there is no computation performed on the fly on incoming data.} The goal is to minimize the total number of rounds.

For problems defined on graphs, the input size $N$ is equal to $n+m$, where $n$ is the number of nodes of the graph
and $m$ is the number of edges. When considering graph problems, in this paper we assume $s = O(n^{1-\epsilon})$.
This regime of memory size, usually called \emph{strongly sublinear memory} regime, is always in compliance
with the aforementioned constraint on the size of the local memory, even when graphs are sparse, for which the
constraint is the most restrictive.

The value of parameter $\epsilon$ can be chosen by the end user. In particular, when solving problem $A$ on input
instance $I$ through a reduction to problem $B$ on input instance $I'$ of increased size, a call to the procedure for
$B$ should set the value of this parameter to a constant $\epsilon' \in (0,1)$ such that $\vert I' \vert^{1-\epsilon'}
= O(\vert I \vert^{1-\epsilon})$.

Since we want to relate the MPC model to classical complexity classes, one must make sure that the model is \emph{uniform},
by which we mean, roughly speaking, that the same algorithm solves the problem for inputs of all (infinitely many) sizes.
Fish et al.~\cite{FishKLRT15} dealt with this issue observing that Karloff et al.'s original definition of the model~\cite{KarloffSV10} 
is non-uniform, allowing it to decide undecidable languages, and thus by reformulating the definition of the model to make
it uniform. Building on that reformulation, and letting $f \colon \mathbb N \to \mathbb R^+$ be a function, we define the
class $\MPC(f(N))$ to be the class of problems solvable in $O(f(N))$ MPC rounds by a uniform family of MPC computations.

\subsection{Circuit Complexity Background}

In this section we review the \emph{Boolean circuit} model of computation. An $n$-input, $m$-output Boolean circuit
$C$ is a directed acyclic graph with $n$ \emph{sources} (i.e., nodes with no incoming edges), called \emph{input nodes},
and $m$ \emph{sinks} (i.e., nodes with no outgoing edges). All non-source nodes are called \emph{gates}, and are labeled
with one among AND, OR, or NOT. The \emph{fan-in} of a gate is the number of its incoming edges. The \emph{size}
of $C$ is the total number of nodes in it. The \emph{depth} of $C$ is the number of nodes in the longest path in $C$.

Note that to decide an entire language, which may contain inputs of arbitrary lengths, we need a \emph{family} of
Boolean circuits, one for each input length. In other words, the Boolean circuit is a natural model for \emph{non-uniform}
computation. When we want to establish relationships between circuit classes and standard machine classes, we need
to define \emph{uniform} circuit classes, with a restriction on how difficult it can be to construct the circuits. The usual
notion of uniformity in this case is that of \emph{logspace-uniformity}: a family of circuits $\{C_n\}_{n \in \mathbb{N}}$
is logspace-uniform if there is an implicitly log-space computable function mapping $1^n$ to the description of the circuit
$C_n$, where implicitly log-space computable means that the mapping can be computed in logarithmic space---see next
section for the definition of logarithmic space.

\begin{definition}[\cite{AroraB09,Sipser97}]\label{def:NC}
For $i \geq 1$, $\NC^i$ is the class of languages that can be decided by a logspace-uniform family of Boolean circuits with
a polynomial number of nodes of fan-in at most two and $O(\log^i n)$ depth. The class $\NC$ is $\cup_{i \geq 1} \NC^i$.
\end{definition}

The complexity classes $\AC^i$ and $\AC = \cup_{i \geq 0} \AC^i$ are defined exactly as $\NC^i$ and $\NC$ except
that gates are allowed to have unbounded fan-in. Hence, for every $i \in \mathbb N$, $\NC^i \subseteq \AC^i$. By
replacing gates with large fan-in by binary trees of gates each with fan-in at most two, we also have $\AC^i \subseteq \NC^{i+1}$.

\subsection{Space Complexity Background}

Space complexity measures the amount of space, or memory, necessary to solve a computational problem.
It serves as a further way of classifying problems according to their computational difficulty, and its study
has a long tradition, which brought several deep and surprising results.

Particularly relevant to this paper are some low-space complexity classes, and specifically, classes of problems
that can be solved with \emph{sublinear} memory. In order for this to make sense---sublinear space is not
even enough to store the input---one must distinguish between the memory used to hold the input and the
working memory, which is the only memory accounted for. Formally, we shall modify the computational model,
introducing a Turing machine with two tapes: a read-only input tape, and a read/write working tape. The first
can only be read, whereas the second may be read and written in the usual way, and only the cells scanned
on the working tape contribute to the space complexity of the computation. Using this two-tape model,
one can define the following complexity classes.

\begin{definition}[\cite{AroraB09,Sipser97}]\label{def:L-NL}
$\L$ is the class of languages that are decidable in logarithmic space on a deterministic Turing machine.
$\NL$ is the class of languages that are decidable in logarithmic space on a nondeterministic Turing machine.
\end{definition}

Informally, logarithmic space is sufficient to hold a constant number of pointers into the input and counters
of $O(\log N)$ bits ($N$ is the length of the input), and a logarithmic number of boolean flags.

As in other standard complexity classes, problems complete for $\L$ or $\NL$ are defined to be the ones
that are, in a certain sense, the most difficult in such classes. To this end, we first need to decide on the kind
of reducibility that would be appropriate. Polynomial-time reducibility would not be very useful because
$\L \subseteq \NL \subseteq \P$, which implies that every language in $\L$ (resp., $\NL$), except $\emptyset$
and $\Sigma^*$, would be $\L$-complete (resp., $\NL$-complete). Hence we need weaker versions of reduction,
ones that involve computations that correspond to sub-classes of $\L$ and $\NL$. One notion of reducibility
that makes sense for the class $\L$ is that of $\NC^1$ reducibility~\cite{CookM87}, where $\NC^1$ is the class
of problems solvable in logarithmic depth by a uniform family of Boolean circuits of bounded fan-in.

\begin{definition}\label{def:L-complete}
A language $B$ is \emph{\L-complete} if (1) $B \in \L$, and (2) every $A$ in $\L$ is $\NC^1$ reducible to $B$.
\end{definition}

$\NC^1$ reducibility has been defined in~\cite{Cook85}. In the literature reductions of even-lower level than
$\NC^1$ are used to identify meaningful notions of $\L$-completeness. Examples are \emph{projections}
and \emph{first-order reductions}. For example, the class first-order logic, denoted as $\FO$, equals the complexity class
$\AC^0$, and since $\AC^0 \subset \NC^1$, a first-order reduction is strictly stronger than an $\NC^1$ reduction.

A good choice for the class $\NL$ is to use log-space reductions, that is, reductions computable by a deterministic Turing
machine using logarithmic space (see~\cite{AroraB09,Sipser97} for a more formal definition of log-space reducibility).

\begin{definition}[\cite{AroraB09,Sipser97}]\label{def:NL-complete}
A language $B$ is \emph{\NL-complete} if (1) $B \in \NL$, and (2) every $A$ in $\NL$ is log-space reducible to $B$.
\end{definition}

Following standard terminology we say that a language is \emph{\L-hard (under $\NC^1$ reductions)}
(resp., \emph{\NL-hard (under log-space reductions)}) if it merely satisfies condition (2) of Definition~\ref{def:L-complete}
(resp., Definition~\ref{def:NL-complete}).

In Appendix~\ref{sec:space} we recall some known results on the space complexity of several fundamental problems.

\section{Massively Parallel Computations and Space Complexity Classes}\label{sec:MPC-Space}

In this section we recall a recent result showing that Boolean circuits can be efficiently simulated in the MPC model,
and then we build on it to derive new results and conjectures.

\subsection{Efficient Circuit Simulation in the MPC Model}

We now recall the main result in~\cite{FreiW19} which, roughly speaking, says that any bounded fan-in Boolean
circuit of depth $d$ and of polynomial size can be simulated in $O(\left\lceil d/\log s \right\rceil)$ MPC rounds.
This result is already implicit in~\cite{RoughgardenVW18}, where it is achieved by a simple simulation whereby
each gate of the circuit is associated with a machine whose responsibility is to compute the output of the gate.
This requires the availability of a number of machines linear in the size of the circuit. Very recently, Frei and
Wada~\cite{FreiW19} came up with a more sophisticated strategy, which uses only a strongly sublinear amount
of machines. Their strategy employs two distinct simulations: for $\NC^1$ circuits they exploit Barrington's
well-known characterization of $\NC^1$ in terms of bounded-width polynomial-size branching programs,
and thus simulate such branching programs in a constant number of rounds; for the higher levels of the $\NC$
hierarchy, the Boolean circuits themselves are directly simulated, suitably dividing the computation into the
simulation of sub-circuits of depth $O(\log n)$, each to be accomplished in $O(1)$ rounds.

The authors work in the original model of Karloff et al.~\cite{KarloffSV10}, but their result seamlessly
applies in the refined MPC model.

\begin{theorem}[\cite{FreiW19}]\label{thm:circuit_simulation}
Let $\DMPC^i$ denote the class of problems solvable by a deterministic MPC algorithm in $O(\log^i N)$ rounds
with $O(N^{1-\epsilon})$ local memory per machine and $O(N^{2(1-\epsilon)})$ total memory. Then,
\[
\NC^{i+1} \subseteq \DMPC^i
\]
for every $i \in \mathcal{N}$ and for every $\epsilon \in (0,1/2)$. (When $i=0$, the result holds also for $\epsilon=1/2$.)
\end{theorem}

Setting $i=0$, we have the following.

\begin{corollary}\label{cor:circuit_simulation}
The class $\NC^1$ can be simulated in $O(1)$ MPC rounds with $O(N^{1-\epsilon})$ local memory per machine
and $O(N^{2(1-\epsilon)})$ total memory, for any constant $\epsilon \in (0,1/2]$.
\end{corollary}

Since $\NC^1 \subseteq \L \subseteq \NL \subseteq \NC^2$ (see, e.g., \cite{Papadimitriou94}), an immediate
by-product of Theorem~\ref{thm:circuit_simulation}  is that some standard space complexity classes can
be efficiently simulated in the MPC model.

\begin{corollary}\label{cor:circuit_simulation2}
The class $\NC^2$, and thus the classes $\L$ and $\NL$, can be simulated in $O(\log N)$ MPC rounds with
$O(N^{1-\epsilon})$ local memory per machine and $O(N^{2(1-\epsilon)})$ total memory, for any constant
$\epsilon \in (0,1/2)$.
\end{corollary}

Corollary~\ref{cor:circuit_simulation2} implies that all the problems discussed in Appendix~\ref{sec:space}
(circuit evaluation, and perfect matching and equivalent problems excluded) can be solved in $O(\log N)$ MPC rounds.

\subsection{New Consequences of Circuit Simulations}

In this section we discuss new consequences of the fact that the MPC model is powerful enough to efficiently
simulate general classes of Boolean circuits.

\begin{theorem}\label{thm:L}
Consider the MPC model where the size of the local memory per machine is $O(N^{1-\epsilon})$ for some
constant $\epsilon \in (0,1/2]$, and assume that $\Omega(N^{2(1-\epsilon)})$ total memory is available.
Let $f \colon \mathbb N \to \mathbb R^+$ be a function. Then, if any $\L$-hard problem can be solved
in $O(f(N))$ MPC rounds, so can all the problems in the class $\L$. Moreover, either all $\L$-complete
problems can be solved in $O(f(N))$ MPC rounds, or none of them can. 
\end{theorem}

\begin{proof}
Both claims follow directly from the definitions of $\L$-hardness and $\L$-completeness, and from
Corollary~\ref{cor:circuit_simulation}. Let $A$ be an $\L$-hard problem that can be solved in $O(f(N))$
MPC rounds. By definition of $\L$-hardness, every problem in $\L$ is $\NC^1$ reducible to $A$. By
assumption, $\epsilon \in (0,1/2]$ and $\Omega(N^{2(1-\epsilon)})$ total memory is available, and thus,
by Corollary~\ref{cor:circuit_simulation}, an $\NC^1$ reduction can be simulated in $O(1)$ MPC rounds,
giving the first claim. Therefore, in particular, if any $\L$-complete problem can be solved in $O(f(N))$
MPC rounds, so can all the other $\L$-complete problems. In other words, either all $\L$-complete
problems can be solved in $O(f(N))$ MPC rounds, or none of them can.
\end{proof}

We remark that in Theorem~\ref{thm:L} no assumption is placed on the function $f(N)$, which therefore
can be of any form, even a constant. Hence, Theorem~\ref{thm:L} says that all the known $\L$-complete
problems such as graph connectivity, graph bipartiteness, cycle detection, and formula evaluation, are
\emph{equivalent} in the MPC model, and in a very strong sense: they all require asymptotically the same number
of rounds. (Analogous equivalences are common in computer science, e.g., in the theory of $\NP$-completeness
and, at a finer-grained level, in the recent fine-grained complexity theory, where equivalence classes of problems
within $\P$, such as the APSP class~\cite{VassilevskaW18,Vassilevska18}, are established.) Thus, this simple 
result provides an explanation for the striking phenomenon that for these well-studied problems we seem unable
to break the $O(\log N)$ barrier in the MPC model. It also implies that the conjectures on the hardness of graph
connectivity and on the hardness of the one cycle vs.\ two cycles problem are equivalent, at least when
$\Omega(N^{2(1-\epsilon)})$ total memory is available.

The next theorem provides an even stronger barrier for improvements in the MPC model.

\begin{theorem}\label{thm:NL}
Consider the MPC model where the size of the local memory per machine is $O(N^{1-\epsilon})$ for some
constant $\epsilon \in (0,1/2]$, and assume that $\Omega(N^{2(1-\epsilon)})$ total memory is available.
Let $f \colon \mathbb N \to \mathbb R^+$ be a function. If any $\L$-hard problem can be solved in
$O(f(N))$ MPC rounds, then either all $\NL$-complete problems can be solved in $O(f(N))$ MPC rounds,
or none of them can. Moreover, if any $\NL$-hard and any $\L$-hard problem can be solved in $O(f(N))$
MPC rounds, so can all the problems in the class $\NL$.
\end{theorem}

\begin{proof}
Let $A$ be an $\L$-hard problem that can be solved in $O(f(N))$ MPC rounds. Then, by Theorem~\ref{thm:L}, 
every problem in the class $\L$ can be solved in $O(f(N))$ MPC rounds and thus, in particular, every log-space
reduction can be computed in $O(f(N))$ MPC rounds. By definition of $\NL$-completeness, every problem
in $\NL$, and thus, in particular, any $\NL$-complete problem, is log-space reducible to any other $\NL$-complete problem,
and this proves the first statement.

Let $B$ be an $\NL$-hard problem that can be solved in $O(f(N))$ MPC rounds. By definition of $\NL$-hardness,
every problem in $\NL$ is log-space reducible to $B$. Since we have just argued that if any $\L$-hard problem
can be solved in $O(f(N))$ MPC rounds, so can any log-space reduction, the second statement follows.
\end{proof}

Once again, we stress that in Theorem~\ref{thm:NL} no assumption is placed on the function $f(N)$,
which therefore can be of any form, even a constant.

Theorem~\ref{thm:NL} indicates that, unless $\L = \NL$, in the MPC model the connectivity problem
on directed graphs, which is both $\NL$-complete and $\L$-hard, is strictly harder than on undirected graphs
in the sense that breaking the current logarithmic barrier, if possible, would be strictly harder.

Notice that we also have the following weaker, but simpler to prove, result: if any problem $\NL$-complete under $\NC^1$
reductions (such as $st$-reachability) can be solved in $O(f(N))$ MPC rounds, so can all the problems in the class $\NL$.
This follows directly from the definition of $\NL$-completeness under $\NC^1$ reductions and from Corollary~\ref{cor:circuit_simulation}.
Notice also that the result in Theorem~\ref{thm:NL} can be extended with the same proof to complexity classes wider than $\NL$,
such as $\NC^2$ or $\P$, for which hardness is defined in terms of log-space reducibility as well.

\subsubsection{New Conjectures}
The common belief that problems such as graph connectivity and list ranking cannot be solved in $o(\log N)$ MPC rounds,
along with the equivalence result of Theorem~\ref{thm:L}, justify the following conjecture.

\begin{conjecture}\label{cnj:L}
No $\L$-hard problem can be solved in $o(\log N)$ MPC rounds with $O(N^{1-\epsilon})$ local memory per machine,
for any constant $\epsilon \in (0,1)$, not even with a polynomial amount of total memory. Equivalently,
\[
\L \nsubseteq \MPC(o(\log N)).
\]
\end{conjecture}

We now show the claimed equivalence.

\begin{proposition}
The two statements in Conjecture~\ref{cnj:L} are equivalent.
\end{proposition}

\begin{proof}
We shall argue that if any of the two statements is wrong, so is the other, and vice versa. Assume $\L \subseteq \MPC(o(\log N))$.
Then, some $\L$-complete, and hence $\L$-hard, problem is contained in $\MPC(o(\log N))$, that is, it can be solved in
$o(\log N)$ MPC rounds. To show the other direction, assume that there exists an $\L$-hard problem that can be solved
in $o(\log N)$ MPC rounds with a polynomial amount of total memory. Then, by Theorem~\ref{thm:L}, every problem
in $\L$ can be solved in $o(\log N)$ MPC rounds, i.e., $\L \subseteq \MPC(o(\log N))$.
\end{proof}

We would like to remark that, in light of Theorem~\ref{thm:L}, Conjecture~\ref{cnj:L} is totally equivalent to the
preceding conjectures on the hardness of graph connectivity or of the one cycle vs.\ two cycles problem~\cite{KarloffSV10,RastogiMCS13,BeameKS17,RoughgardenVW18,YaroslavtsevV18};
however, Theorem~\ref{thm:L} significantly strengthens the evidence for such conjectures.

Likewise, Theorem~\ref{thm:NL} provides a justification for the following conjecture.

\begin{conjecture}\label{cnj:NL}
No $\NL$-hard and $\L$-hard problem can be solved in $o(\log N)$ MPC rounds with $O(N^{1-\epsilon})$ local memory
per machine, for any constant $\epsilon \in (0,1)$, not even with a polynomial amount of total memory. Equivalently,
\[
\NL \nsubseteq \MPC(o(\log N)).
\]
\end{conjecture}

We now show the claimed equivalence.

\begin{proposition}
The two statements in Conjecture~\ref{cnj:NL} are equivalent.
\end{proposition}

\begin{proof}
We shall argue that if any of the two statements is wrong, so is the other, and vice versa. Assume $\NL \subseteq \MPC(o(\log N))$.
Then, in particular, $st$-reachability can be solved in $o(\log N)$ MPC rounds. Since $st$-reachability is both $\NL$-hard
and $\L$-hard, this contradicts the first statement. To show the other direction, assume that there exists an $\NL$-hard
and $\L$-hard problem that can be solved in $o(\log N)$ MPC rounds with a polynomial amount of total memory. Then, by
Theorem~\ref{thm:NL}, every problem in $\NL$ can be solved in $o(\log N)$ MPC rounds, i.e., $\NL \subseteq \MPC(o(\log N))$.
\end{proof}

Figure~\ref{fig:conjectures} depicts the conjectured relationships among $\L$, $\NL$, and $\MPC(o(\log N))$.
Observe that since $\L \subseteq \NL$, Conjecture~\ref{cnj:L} implies Conjecture~\ref{cnj:NL}. Hence, unless $\L = \NL$,
Conjecture~\ref{cnj:NL} is weaker than Conjecture~\ref{cnj:L}, and thus more likely to be true.

\begin{figure}
    \begin{center}
       \includegraphics[width=0.65\textwidth]{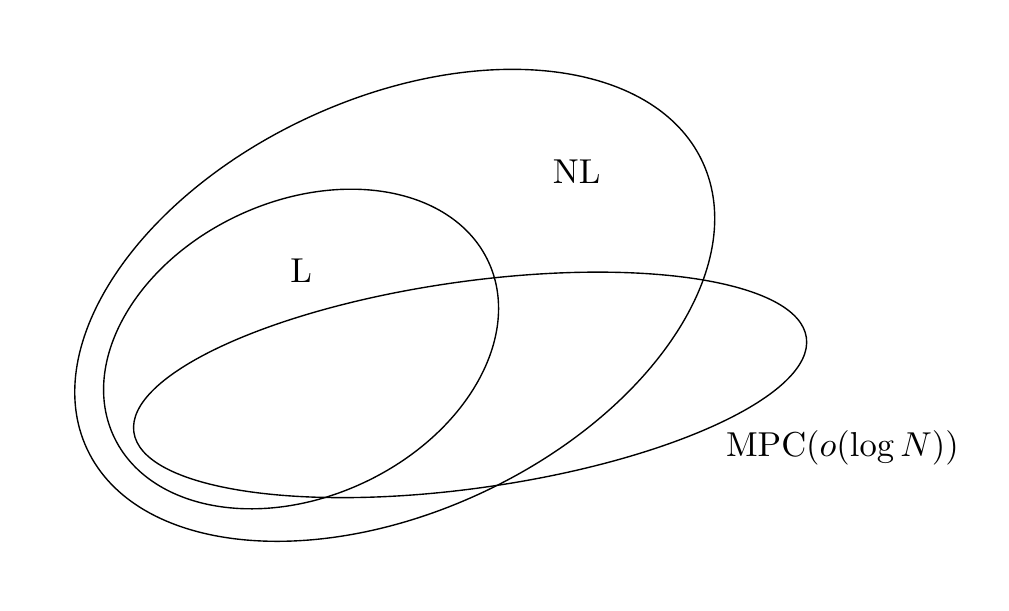}
       \caption{Conjectured relationships among classes $\L$, $\NL$, and $\MPC(o(\log N))$.}\label{fig:conjectures}
    \end{center}
\end{figure}

We stress that breaking either conjecture would have vast consequences because of the large number of fundamental
problems contained in $\L$ and $\NL$. This is somewhat in contrast, e.g., to the Strong Exponential Time Hypothesis
(SETH), a popular hardness assumption on the complexity of $k$-SAT used to prove a plethora of conditional lower
bounds, especially in the realm of polynomial-time algorithms~\cite{Vassilevska18}, whose refutation would have
more limited algorithmic consequences.

These two conjectures can be used as a base for conditional lower bounds in the MPC model, in the same
way as the one cycle vs.\ two cycles conjecture was used as a hardness assumption
in~\cite{AndoniNOY14,YaroslavtsevV18,AndoniSZ19,LackiMOS,GhaffariKU19,BehnezhadDELM19}.

\section{Reductions and Equivalences in Massively Parallel Computations}\label{sec:reductions}

In this section we discuss two equivalence classes of problems and some conditional lower bounds in the MPC model.
The two equivalence classes both contain problems equivalent to each other under $O(1)$-round MPC reductions
and for which the best known upper bound is $O(\log N)$ rounds, but differ in terms of the low-space computational
complexity characterization of the problems they contain.

As a consequence of the results of Section~\ref{sec:MPC-Space}, most of these reductions and  equivalences follow from
known hardness and completeness results for low-space complexity classes such as $\L$ and $\NL$.

We will also show novel reductions and equivalences in the MPC model. Some of such reductions crucially require the
availability of up to polynomially many machines (equivalently, a total amount of memory up to polynomial in the input size),
which are used to host up to a polynomial number of copies of the input data. The quick creation of so many input replicas
can be achieved through the use of a simple two-step broadcast procedure, as shown in the following lemma.

\begin{lemma}\label{lem:replication}
The input data can be replicated up to a polynomial number of times in $O(1)$ MPC rounds.
\end{lemma}

\begin{proof}
Assume, without loss of generality, that initially all the input data is held by the first $\beta$ consecutively numbered machines.
We use a basic two-step broadcast procedure to replicate the contents of each machine in $O(N^{1-\epsilon})$ other machines.
Any polynomial-factor replication can thus be achieved by repeating the procedure a constant amount of times.

Let $c$ be a sufficiently large positive constant. For each $i \in [\beta]$, machine $i$ logically partitions its memory contents in
$cN^{1-\epsilon}$ parts, one for each word. Then, machine $i$ sends the $j$-th word to the $(\beta + (i-1)cN^{1-\epsilon} + j)$-th
machine. Finally, each of these machines broadcasts the word received from machine $i$ to all the other machines in the range
$\beta + (i-1)cN^{1-\epsilon} + 1, \dots, \beta + icN^{1-\epsilon}$, thus yielding a factor-$cN^{1-\epsilon}$ replication
of the contents of each machine.
\end{proof}

\subsection{An Equivalence Class for Undirected Graph Connectivity}

In this section we discuss the MPC equivalence class for graph connectivity in undirected graphs. This problem, which
asks to determine whether a given undirected graph is connected or not, was one of the first problems to be shown
$\L$-hard under (uniform) $\NC^1$ reductions~\cite{CookM87}, and then it was placed in $\L$ by the remarkable
algorithm of Reingold~\cite{Reingold08}. Exploiting the results of Section~\ref{sec:MPC-Space}, we know that one can
recycle all the reductions that have been developed in classical complexity theory for showing hardness and completeness
for class $\L$ in the MPC model as well, since these can all be simulated in $O(1)$ MPC rounds with $O(N^{2(1-\epsilon)})$
total memory. This immediately implies that the class of $\L$-complete problems forms an equivalence class
in the MPC model as well. Specifically, for example, either all the following problems can be solved with a sublogarithmic
MPC algorithm, or none of them can:\footnote{See Appendix~\ref{sec:space} for precise definitions of the various problems,
as well as for references.}
graph connectivity,
connectivity for promise graphs that are a disjoint union of cycles,
$st$-connectivity,
$st$-reachability for directed graphs of out-degree one,
cycle detection,
order between vertices,
formula evaluation,
and planarity testing.

\paragraph{Recycling (some) old $\SL$-completeness results.}
Many more problems can be placed in this MPC equivalence class almost effortlessly: this is the case for some problems
complete for the class \emph{symmetric logarithmic space} ($\SL$), a class defined by Lewis and Papadimitriou~\cite{LewisP82}
to capture the complexity of undirected $st$-connectivity before this was eventually settled by the breakthrough of Reingold.
Completeness in $\SL$ is defined in terms of log-space reductions, and $st$-connectivity is one complete problem for it.
Since $\L \subseteq \SL$, Reingold's algorithm made these two classes collapse, thus widening the class $\L$ with many
new problems. However, completeness for $\SL$ does not translate into completeness for $\L$, since the latter is defined
in terms of a lower-level kind of reduction. Luckily, some of the log-space reductions devised to show hardness for $\SL$
turn out to be actually stronger than log-space. This is the case, e.g., of testing whether a given graph is bipartite (or,
equivalently, $2$-colorable), as we show next.

\begin{lemma}
Graph bipartiteness is equivalent to $st$-connectivity under $O(1)$-round MPC reductions, with $O(n^{1-\epsilon})$
local memory per machine for some constant $\epsilon \in (0,1)$, and $O(n(n+m))$ total memory.
\end{lemma}

\begin{proof}
Jones et al.~\cite{JonesLL76} showed that testing whether a graph is non-bipartite is equivalent to $st$-connectivity under
log-space reductions. We will now argue that both reductions can be simulated in $O(1)$ MPC rounds.

We start by showing that $st$-connectivity reduces to graph bipartiteness in $O(1)$ MPC rounds. The idea is to make use
of the fact that a graph is bipartite if and only if it has no cycle of odd length. Given an instance $G = (V,E)$, $s$, and $t$
of $st$-connectivity, we build a new graph $G' = (V',E')$ where
\[
V' = \{u, u' : u \in V\} \cup \{e, e' : e \in E\} \cup \{w : w \notin V \cup E\}
\]
and
\[
E' = \{\{u, e\}, \{e, v\}, \{u', e'\}, \{e', v'\} : e = \{u, v\} \in E\} \cup \{\{s, s'\}, \{t, w\}, \{t', w\}\}.
\]
Then observe that $G'$ contains an odd length cycle if and only if $s$ is connected to $t$ in $G$.
Nodes and edges of $G'$ can be easily generated in $O(1)$ rounds, and stored with $O(n(n+m))$ total memory.
Since $\vert V' \vert = O(n^2)$, when working with $G'$ the size of the local memory is set to $n^{2(1-\epsilon')}$
where $\epsilon' \in (0,1)$ is a constant such that $n^{2(1-\epsilon')} = O(n^{1-\epsilon})$. Then, an $O(f(n))$-round
algorithm for graph bipartiteness translates into an $O(f(n))$-round algorithm for $st$-connectivity.

We now show that graph bipartiteness reduces to $st$-connectivity in $O(1)$ MPC rounds. Given an instance $G = (V,E)$,
the idea is to construct a new graph by creating two copies of each node, call them copy 0 and copy 1, and then
for any edge ${u, v} \in E$, connecting the 0 copy of $u$ to the 1 copy of $v$ and vice versa. This can be trivially done
in $O(1)$ MPC rounds. This new graph, $G'$, is not bipartite if and only if there is some node $w$ such that the 0 copy
of $w$ is reachable from the 1 copy of $w$. To take care of the phrase ``there is some node $w$'', $n$ copies of $G'$ are
created and new nodes $s$ and $t$ are introduced. Then $s$ (resp., $t$) is connected to the 0 (resp., 1) copy of the $i$-th
node in copy $i$. By Lemma~\ref{lem:replication}, this can be accomplished in $O(1)$ MPC rounds as well.
\end{proof}

A good source of problems complete for $\SL$ is~\cite{AlvarezG00}.

\paragraph{From decision to non-decision problems.}
Complexity classes such as $\L$ contain problems phrased as decision problems. Nevertheless, it is often easy to
transform them into their non-decision version. As an example, consider \emph{order between vertices} (ORD). ORD is the decision
version of list ranking, the problem of obtaining a total ordering from a given successor relation~\cite{Etessami97}.
It is easy to argue the following equivalence.

\begin{lemma}
List ranking is equivalent to order between vertices under $O(1)$-round MPC reductions, with $O(n^{1-\epsilon})$
local memory per machine for some constant $\epsilon \in (0,1)$, and $O(n^3)$ total memory.
\end{lemma}

\begin{proof}
Order between vertices trivially reduces to list ranking. We now argue that list ranking is reducible under $O(1)$-round MPC
reductions to ORD when there are polynomially many available machines. The reduction is as follows: (1) create $\binom{n}{2}$
replicas of the $n$ inputs across the machines; by Lemma~\ref{lem:replication} this takes $O(1)$ MPC rounds; (2) in parallel,
solve ORD for each pair of nodes, one pair for each input replica; (3) each of $n$ designated machines outputs the rank of a
distinct node $u$ by counting the number of yes/no outputs for ORD for the pair $(u,v)$, for each $v \neq u$: doing this is
tantamount to doing summation, which can be done in $O(1)$ MPC rounds by~\cite{ChandraSV84} and Corollary~\ref{cor:circuit_simulation}.
\end{proof}

\paragraph{Non-pairwise reductions.}
Sometimes back-and-forth reductions between two problems are not known. In this case their equivalence may nevertheless
be established through a series of reductions involving related problems. As an example, we now show that a bunch of problems
related to graph connectivity are all equivalent under $O(1)$-round MPC reductions. Besides graph connectivity and $st$-connectivity,
these are determining the connected components of an undirected graph, counting the number of connected components (\# connected 
components), finding a minimum-weight spanning forest (MSF), and finding a minimum cut. See Figure~\ref{fig:connectivity}. Recall
that a connected component of an undirected graph is a maximal set of nodes such that each pair of nodes is connected by a path,
and it is usually represented by a labeling of nodes such that two nodes have the same label if and only if they are in the same
connected component. A minimum spanning forest of a weighted graph is the union of the minimum spanning trees for its
connected components. In the minimum cut problem we have to find a partition of the nodes of a graph into two disjoint sets
$V_1, V_2 = V \setminus V_1$ such that the set of edges that have exactly one endpoint in $V_1$ and exactly one endpoint in
$V_2$ is as small as possible. 

\begin{figure}
    \begin{center}
       \includegraphics[width=\textwidth]{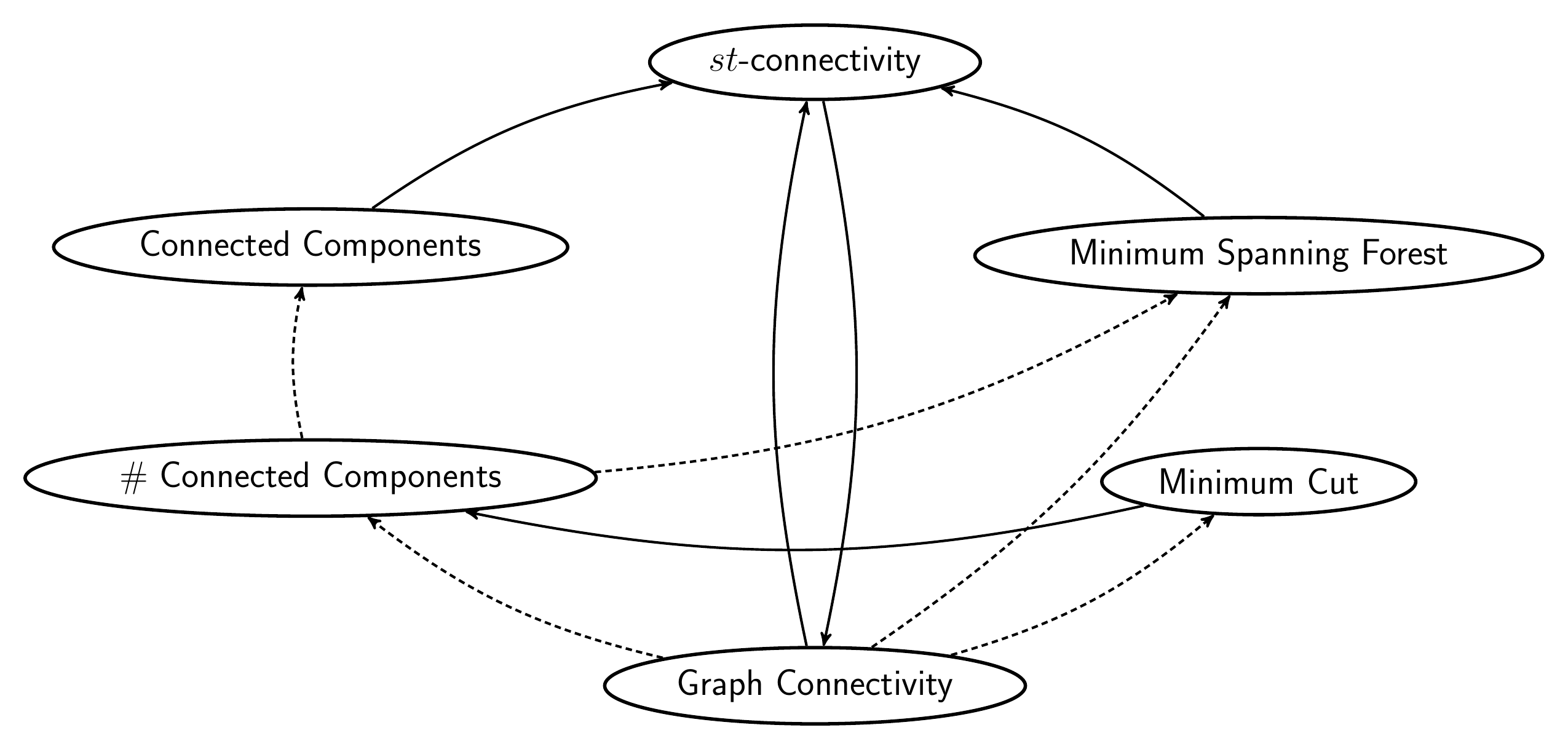}
       \caption{Constant-round reductions among graph connectivity and related problems. Dashed arrows correspond to trivial reductions.}\label{fig:connectivity}
    \end{center}
\end{figure}

\begin{lemma}
Graph connectivity, $st$-connectivity, \# connected components, connected components, minimum spanning forest,
and minimum cut are all equivalent under $O(1)$-round MPC reductions, with $O(n^{1-\epsilon})$ local memory per
machine for some constant $\epsilon \in (0,1)$, and $\tilde O(n^2 m (n+m))$ total memory.
\end{lemma}

\begin{proof}
The reductions from graph connectivity to detecting the number of connected components, to MSF, and to minimum cut
are obvious. The reductions from \# connected components to connected components and to MSF are also obvious.
We already mentioned that there is a non-obvious low-level equivalence between graph connectivity and $st$-connectivity,
shown by Chandra et al.~\cite{ChandraSV84}.

Log-space reductions to $st$-connectivity from MSF and from connected components were given by Nisan and
Ta-Shma~\cite{NisanT95} for showing that the class $\SL$ is closed under complement. Here we will argue that
these reductions can be simulated in $O(1)$ MPC rounds.

We first discuss how to reduce connected components to $st$-connectivity. The reduction is to evaluate, for each
node $s$, $st$-connectivity for every other node $t$ in the graph. Then, the label $\ell$ assigned to node $s$ is
\[
\ell(s) = \min_{t \in V} \{\textrm{ID of node $t$ such that $t$ is connected to $s$}\}.
\]
By Lemma~\ref{lem:replication} and by the fact that the min function can be evaluated in $O(1)$ rounds
(by~\cite{ChandraSV84} and Corollary~\ref{cor:circuit_simulation}), this reduction can be accomplished in $O(1)$ MPC rounds.

We now discuss how to reduce MSF to $st$-connectivity. This is based on the following simple property shown in~\cite{NisanT95},
implicitly used in several other similar reductions~\cite{AhnGM12,HegemanPPSS15,PanduranganRS18,AndoniSSWZ18}:
an edge $e = \{u, v\}$ is in the minimum-weight spanning forest if and only if $u$ is not connected to $v$ in the graph
made up of all edges having lower weight than $e$. Then, by Lemma~\ref{lem:replication}, in $O(1)$ rounds the input
graph can be replicated $m$ times across the available machines, and then testing whether a designated edge $e$
is in the (unique) minimum-weight spanning forest of $G$ can be done in parallel for each edge of the graph.

Finally, we discuss how to reduce minimum cut to \# connected components. This is based on the parallelization of Karger's
celebrated contraction algorithm~\cite{Karger93}. Recall that Karger's algorithm repeats $O(n^2 \log n)$ times the process
of contracting randomly chosen edges, one by one, until only two nodes remain. By assumption we have enough machines
to replicate the input graph that many times in $O(1)$ MPC rounds (by Lemma~\ref{lem:replication}) and run the $O(n^2 \log n)$
trials in parallel. Identifying the minimum cut from these results can be done in $O(1)$ MPC rounds.

The question is therefore how to run a single time the contraction algorithm. To this end, it is convenient to work
with the equivalent reformulation of the contraction algorithm whereby we first generate a random permutation
of the $m$ edges, and then contract edges in the order in which they appear in the permutation. Generating a random
permutation can be done in $O(1)$ rounds by having each processor take one edge and assign it a score chosen
uniformly at random from a sufficiently large range of integers, and then by sorting these scores. Then, consider
any such permutation. The key property is that it has a prefix such that the set of edges in this prefix induces two
connected components (the two sides of the cut), that any prefix which is too short yields more than two connected
components, and that any prefix which is too long yields only one. Hence, with enough machines available, we can
determine the correct prefix by examining all the $m$ prefixes of each permutation in parallel.
\end{proof}

We can now summarize all the results of this section.

\begin{theorem}\label{thm:conn-equivalence}
The following problems are all equivalent under $O(1)$-round MPC reductions, with $O(n^{1-\epsilon})$ local memory
per machine for some constant $\epsilon \in (0,1)$, and $\tilde O(n^2 m (n+m))$ total memory: graph connectivity, connectivity
for promise graphs that are a disjoint union of cycles, $st$-connectivity, $st$-reachability for directed graphs of out-degree
one, cycle detection, order between vertices, formula evaluation, planarity testing, graph bipartiteness, list ranking, \# connected
components, connected components, minimum spanning forest, and minimum cut.
\end{theorem}

\paragraph{Conditional hardness: $\L$-hard problems.}
Finally, there are problems known to be $\L$-hard, but not known to be in $\L$, such as densest subgraph and perfect matching.
Since for these problems only one-way reductions from problems in $\L$ are known, we don't know whether they are part of the
equivalence class of undirected graph connectivity.

\subsection{An Equivalence Class for Directed Graph Connectivity}

In this section we discuss the MPC equivalence class for graph connectivity in directed graphs. The problem corresponding
to $st$-connectivity in directed graphs is \emph{$st$-reachability}, that is, the problem of detecting whether there is a path
from a distinguished node $s$ to a distinguished node $t$ in a directed graph. $st$-reachability is the prototypical complete
problem for $\NL$~\cite{Papadimitriou94,Sipser97,AroraB09}.

By Definition~\ref{def:NL-complete}, hardness in class $\NL$ is defined with respect to log-space reducibility, but we do not
know whether log-space computations can be simulated in $o(\log N)$ MPC rounds---in fact, in Section~\ref{sec:MPC-Space}
we conjecture they cannot. However, it turns out that many of the known log-space reductions that establish $\NL$-hardness
of problems can be simulated in $O(1)$ MPC rounds. This is the case, for example, of the reductions between
$st$-reachability and \emph{shortest path}, the other canonical example of $\NL$-complete problem which, given an
undirected (unweighted) graph, two distinguished nodes $s$ and $t$, and an integer $k$, asks to determine if the length
of a shortest path from $s$ to $t$ is $k$.

\begin{lemma}
Shortest path on unweighted graphs is equivalent to $st$-reachability under $O(1)$-round MPC reductions, with
$O(n^{1-\epsilon})$ local memory per machine for some constant $\epsilon \in (0,1)$, and $O(n(n+m))$ total memory.
\end{lemma}

\begin{proof}
We first show that $st$-reachability can be reduced to shortest path in $O(1)$ MPC rounds. For integer $k$ we denote
the integers $\{1,2,\dots,k\}$ by $[k]$. Given a directed graph $G=(V,E)$ and two designated nodes $s$ and $t$,
we create a new (undirected) layered graph $G'=(V',E')$ where
\[
V' = \{v_i : v \in V, i \in [n]\}
\]
and
\[
E' = \{ \{v_i,v_{i+1}\} : v \in V, i \in [n-1]\} \cup \{ \{u_i,v_{i+1}\} : (u,v) \in E, i \in [n-1] \}.
\]
It is easy to see that there is a directed path from $s$ to $t$ in $G$ if and only if there is a path of length $n-1$
from $s_1$ to $t_n$ in $G'$.

We now show the other direction. Given an undirected graph $G=(V,E)$, two designated nodes $s$ and $t$, and an
integer $b \in [n-1]$, we create a new directed layered graph $G'=(V',E')$ where
\[
V' = \{v_i : v \in V, i \in [b]\}
\]
and
\[
E' = \{ (v_i,v_{i+1}) : v \in V, i \in [b-1]\} \cup \{ (u_i,v_{i+1}) : \{u,v\} \in E, i \in [b-1] \}.
\]
Then again it is easy to see that the length of a shortest path from $s$ to $t$ is at most $b$ if and only if there is
a directed path from $s_1$ to $t_b$ in $G'$. If the length is at most $b$ then one can determine if it is exactly $b$
by repeating the same construction with $b-1$ in place of $b$.

In both directions, nodes and edges of $G'$ can be easily generated in $O(1)$ rounds, and stored with $O(n(n+m))$
total memory. Since $\vert V' \vert \leq n^2$, when working with $G'$ the size of the local memory is set to $n^{2(1-\epsilon')}$
where $\epsilon' \in (0,1)$ is a constant such that $n^{2(1-\epsilon')} = O(n^{1-\epsilon})$. Then, an $O(f(n))$-round
algorithm for one problem translates into an $O(f(n))$-round algorithm for the other, and vice versa.
\end{proof}

There are other $\NL$-complete problems that can be shown to be equivalent under $O(1)$-round MPC reductions.
Some examples are directed cycle detection, by a simple adaptation of the preceding reductions, and strong
connectivity, which follows from a result in~\cite{ChandraSV84}. We suspect that many other log-space reductions
are actually (or can easily be translated into) $O(1)$-round MPC reductions, thus enabling us to enlarge the equivalence
class for graph connectivity in directed graphs almost effortlessly by leveraging known results in complexity theory.

When this is not possible, one might have to devise novel reductions. We now do so for some important shortest-path-related
problems as well as for some graph centrality problems.

\subsubsection{New Fine-Grained MPC Reductions: Constant-Round Equivalences Between Graph Centrality Problems, APSP, and Diameter}

In this section we prove a collection of constant-round equivalences between shortest path and many other
problems on weighted graphs.

First, some preliminaries. In a graph problem, the input is an $n$-node $m$-edge (directed or undirected) graph
$G=(V,E)$ with integer edge weights $w \colon E \to \{-M,\dots,M\}$ where $M = O(n^c)$ for some positive
constant $c$. $G$ is assumed to contain no negative-weight cycles. Let $d(u,v)$ denote the (shortest-path) distance
from node $u \in V$ to node $v \in V$, that is, the minimum over all paths from $u$ to $v$ of the total weight sum
of the edges of the path. If there is no path connecting the two nodes, i.e., if they belong to different connected
components, then conventionally the distance is defined to be infinite.

The fundamental all-pairs shortest paths (APSP) problem is to compute $d(u,v)$ for every pair of nodes $u,v \in V$.
In the (sequential) RAM model, APSP has long been known to admit an $O(n^3)$ time algorithm. Despite the long
history, no algorithm that runs in time $O(n^{3-\epsilon})$ for some constant $\epsilon > 0$ is known, and it is
conjectured that no such algorithm exists~\cite{VassilevskaW18,Vassilevska18}. This conjecture is commonly used
as a hardness hypothesis in fine-grained complexity theory to rule out faster algorithms than those currently known
for several problems~\cite{Vassilevska18}. Beyond such APSP-hardness results, some important problems have been
shown to be \emph{equivalent} to APSP, in the sense that either all such problems admit $O(n^{3-\epsilon})$ time
algorithms, or none of them do~\cite{VassilevskaW18,AbboudGW15,Vassilevska18}.

These equivalences and most hardness results under the APSP hypothesis rely on a reduction from APSP to
the \emph{negative triangle} problem, which asks whether a graph has a triangle with negative total weight.
Although negative triangle can be easily solved in $O(1)$ MPC rounds thanks to Lemma~\ref{lem:replication}, a
key building block in the reduction from APSP~\cite{VassilevskaW18} is a well-known equivalence~\cite{FischerM71}
between APSP and the \emph{distance product} problem of computing the product of two matrices over the $(\min,+)$
semiring (also known as \emph{min-plus matrix multiplication}); unfortunately, in the reduction from APSP to distance
product there are $\lceil \log n \rceil$ of such matrix products (by using the ``repeated squaring'' strategy), and this
takes $O(\log n)$ MPC rounds---which is likely to be best possible, for a reason that will be clear in the next paragraph.
Hence in the MPC model we cannot rely on a reduction to negative triangle to prove equivalences to APSP or
related hardness results: we need sublogarithmic fine-grained reductions.

Hence we shall follow a different path, by reducing from the shortest path problem. Given a weighted graph, two
distinguished nodes $s$ and $t$, and an integer $k$, shortest path is the problem of determining if the distance
of a shortest path from $s$ to $t$ is $k$. This problem is $\NL$-complete, even for undirected and unweighted
graphs~\cite{BorodinCDRT89}. (This also explains why the repeated matrix squaring discussed in the previous
paragraph is best possible under Conjecture~\ref{cnj:NL}.) As we will show shortly, it turns out that shortest
path is reducible in $O(1)$ MPC rounds to several fundamental graph problems, including many graph centrality
problems defined in terms of shortest paths. Then, by crucially exploiting the availability of many machines,
we will argue that APSP is $O(1)$-round reducible to shortest path. Obvious reductions to APSP complete the
picture and establish the equivalence of all these problems under $O(1)$-round MPC reductions. See
Figure~\ref{fig:summary} for a complete summary.

\begin{figure}[h]
    \begin{center}
       \includegraphics[width=\textwidth]{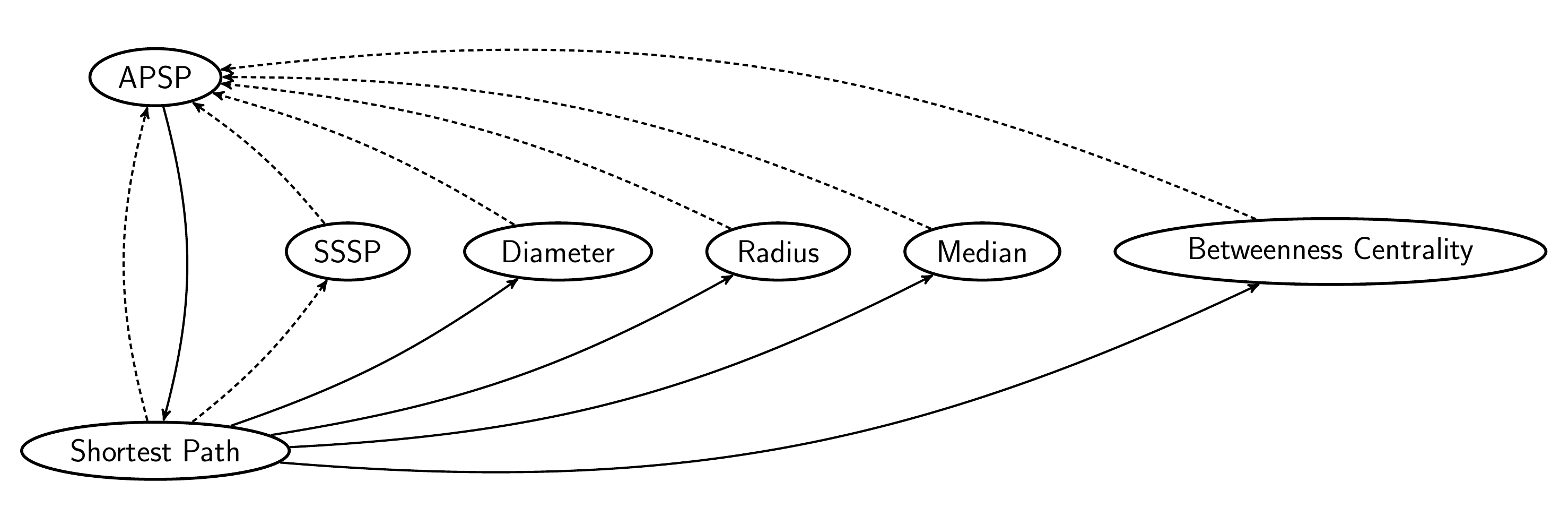}
       \caption{The constant-round reductions shown in this section. Dashed arrows correspond to trivial reductions.}\label{fig:summary}
    \end{center}
\end{figure}

We now formally define the problems we are going to investigate. The \emph{eccentricity} $\epsilon(v)$ of a node
$v$ is the greatest distance between $v$ and any other node. It can be thought of as how far a node is from
the node most distant from it in the graph. The \emph{diameter} of a graph is the greatest distance between
any pair of nodes or, equivalently, the maximum eccentricity of any node in the graph, that is,
\[
\textrm{diam}(G) = \max_{u \in V} \max_{v \in V} d(u,v).
\]
The \emph{radius} of a graph is the minimum eccentricity of any node, that is,
\[
\textrm{radius}(G) = \min_{u \in V} \max_{v \in V} d(u,v),
\]
and a node with minimum eccentricity is called a \emph{center} of the graph.
The \emph{distance sum} of a node $u$ is the sum of the distances from $u$ to all the other nodes,
that is, $\sum_{v \in V} d(u,v)$.\footnote{This is sometimes also called the \emph{farness} of $u$.}
In a (strongly) connected graph, the \emph{closeness centrality} of a node $u$ is the normalized inverse
of its distance sum, that is,
\[
\text{CC}(u) = \frac{n-1}{\sum_{v \in V} d(u,v)}.
\]
A node with maximum closeness centrality, i.e., a node that minimizes the sum of the distances to all other nodes
is called a \emph{median} of the graph, and the value
\[
\min_{u \in V} \sum_{v \in V} d(u,v)
\]
is defined as the \emph{median} of the graph.
The \emph{betweenness centrality} of a node $u$ is defined as
\[
\text{BC}(u) = \sum_{s,t \in V \setminus \{u\}, s \neq t} \dfrac{\sigma_{s,t}(u)}{\sigma_{s,t}},
\]
where $\sigma_{s,t}$ is the total number of distinct shortest paths from $s$ to $t$, and $\sigma_{s,t}(u)$ is the number
of such paths that use $u$ as an intermediate node. Informally, betweenness centrality measures the propensity of a node
to be involved in shortest paths.

We start by showing the simple fine-grained equivalence between APSP and shortest path.

\begin{lemma}\label{lem:APSP-SP}
APSP is equivalent to shortest path under $O(1)$-round MPC reductions, with $O(n^{1-\epsilon})$ local memory
per machine for some constant $\epsilon \in (0,1)$, and $O(n^2 (n+m))$ total memory.
\end{lemma}

\begin{proof}
The reduction from shortest path to APSP is obvious. The other direction is also immediate when we have enough
machines, and specifically $O(n^2 (n+m))$ total memory: by Lemma~\ref{lem:replication} we can create $2\binom{n}{2}$
copies of the input graph in $O(1)$ MPC rounds, and then in parallel, one pair for each copy, compute the shortest path
for each (ordered, if the graph is directed) pair of nodes.
\end{proof}

In the following results we will use roughly the same reduction. We start with the problem of determining the diameter of a graph.

\begin{lemma}\label{lem:SP-Diam}
Shortest path is $O(1)$-round MPC reducible to diameter, with $O(n^{1-\epsilon})$ local memory per machine
for some constant $\epsilon \in (0,1)$, and $O(n+m)$ total memory.
\end{lemma}

\begin{proof}
We start with the case of undirected graphs. Given an instance of shortest path, the idea is to alter the input graph
by sticking two new and sufficiently long paths to nodes $s$ and $t$, so that the path of largest total weight includes
both $s$ and $t$.

This is sufficient if the original graph $G$ is connected; otherwise, the diameter is infinite, and from this information
we cannot determine the length of a shortest path from $s$ to $t$. Hence, we shall first make $G$ connected in a way
that alters the distance between $s$ and $t$ only if they are not connected in $G$. Since the distance between any
two nodes can be at most $(n-1)M$, this can be achieved by adding to the graph a new node $v$ and $n$ edges
of weight $nM$ between $v$ and any other node. Then, we append two additional chains to $s$ and $t$, each with
$2n$ edges of weight $M$, and denote this modified graph by $G'$. See Figure~\ref{fig:SP-Diam}.

\begin{figure}[h]
    \begin{center}
       \includegraphics[trim=0cm 0cm 0cm 0cm, clip=true, width=0.9\textwidth]{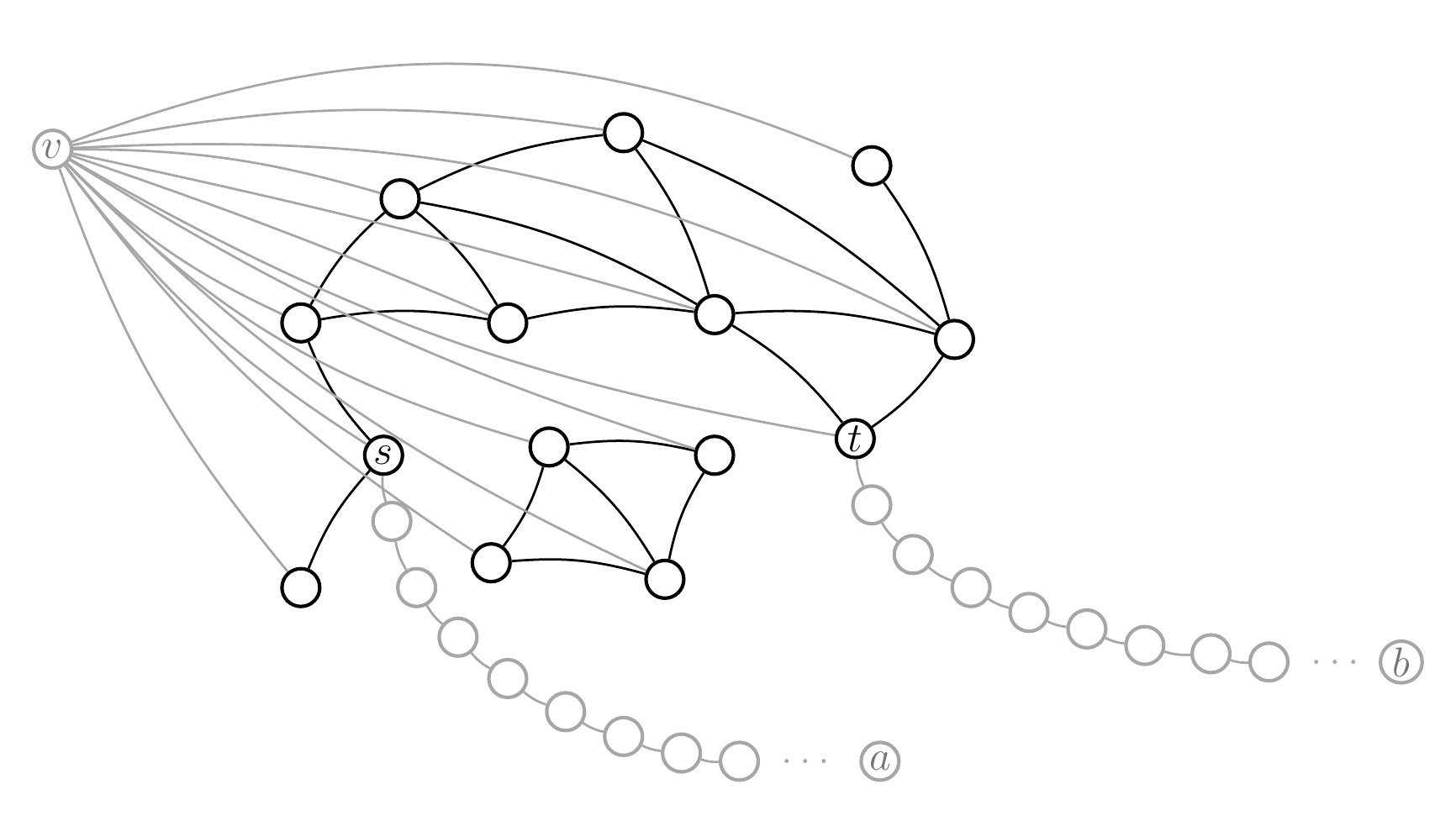}
       \caption{Reduction from shortest path to diameter. Nodes and edges of the original graph $G$ are in black,
       whereas nodes and edges added in the reduction are in gray.}\label{fig:SP-Diam}
    \end{center}
\end{figure}

This reduction can be performed in $O(1)$ MPC rounds, it increases the number of nodes and the number
of edges by $O(n)$, and the maximum absolute weight by a factor of $O(n)$. Therefore, any MPC algorithms that
runs in $O(f(n,m,M))$ rounds in the new graph $G'$ can be used to solve the original instance $G$ in $O(f(n,m,M))$
rounds as well.

Observe that the diameter of the modified graph $G'$ must include the two chains appended to $s$ and $t$.
Hence any algorithm for the diameter when executed on graph $G'$ always returns $4nM$ plus the shortest-path distance
between $s$ and $t$ in $G'$. By construction, the latter quantity, which we denote by $\alpha$, is at most $(n-1)M$ if $s$
and $t$ are connected in $G$, and (exactly) $2nM$ otherwise. Thus the answer to shortest path is $\alpha$ if the
diameter of $G'$ is at most $4nM + (n-1)M$, and infinity otherwise.

In the directed case, we use the same weighted graph $G'$ as before, adding one parallel edge for each edge,
both with the same weight but with opposite directions. The rest of the algorithm is the same and its analysis is
analogous to the undirected case.
\end{proof}

Observe that $st$-connectivity in undirected or directed graphs can also be reduced to diameter, with the same reduction.
However, in undirected graphs $st$-connectivity is only $\L$-hard, while shortest path is $\NL$-hard.

\begin{lemma}\label{lem:SP-Radius}
Shortest path is $O(1)$-round MPC reducible to radius, with $O(n^{1-\epsilon})$ local memory per machine
for some constant $\epsilon \in (0,1)$, and $O(n+m)$ total memory.
\end{lemma}

\begin{proof}
We start with the case of undirected graphs. Given an instance of shortest path, we will construct the graph $G'$
of Figure~\ref{fig:SP-Diam} used in the reduction from shortest path to diameter, and then we will modify $G'$
to obtain a new graph $G''$ such that $\text{radius}(G'') = \text{diameter}(G')$.

The graph $G''$ is obtained from $G'$ by creating a second copy of it, and then by contracting node $b$ of the
first copy of $G'$ with node $a$ of the second copy of $G'$. (Recall that the contraction of a pair of nodes $v_i$
and $v_j$ of a graph produces a graph in which the two nodes $v_1$ and $v_2$ are replaced with a single node
$v$ such that $v$ is adjacent to the union of the nodes to which $v_1$ and $v_2$ were originally adjacent.)
This reduction can be performed in $O(1)$ MPC rounds, it increases the number of nodes by $O(n)$ and the
number of edges by $O(m)$, and the maximum absolute weight by a factor of $O(n)$.

Let $c$ be the node resulting from this contraction. It is easy to see that $c$ is the center of this newly constructed
graph $G''$: in fact, by the symmetry of $G''$ and by the assignment of the edge weights, any other node has higher
eccentricity. Thus, $\text{radius}(G'') = \text{diameter}(G')$, and hence we can proceed as in the proof of
Lemma~\ref{lem:SP-Diam}.

In the directed case, we use the same weighted graph $G''$ as before, adding one parallel edge for each edge,
both with the same weight but with opposite directions. The rest of the algorithm is the same and its analysis is
analogous to the undirected case.
\end{proof}

\begin{lemma}\label{lem:SP-Median}
Shortest path is $O(1)$-round MPC reducible to median, with $O(n^{1-\epsilon})$ local memory per machine
for some constant $\epsilon \in (0,1)$, and $O(n+m)$ total memory.
\end{lemma}

\begin{proof}
We start with the case of undirected graphs. Given an instance of shortest path, we construct the graph $G''$
as in the reduction from shortest path to radius (see proof of Lemma~\ref{lem:SP-Radius}), and compute
$\text{median}(G'')$. Then, we shall edit $G''$ by adding two nodes, $a'$ and $b'$, as well as two edges,
$\{a,a'\}$ and $\{b,b'\}$, both of weight $M$. We call the resulting graph $G'''$. This reduction can be performed
in $O(1)$ MPC rounds, it increases the number of nodes by $O(n)$ and the number of edges by $O(m)$, and the
maximum absolute weight by a factor of $O(n)$.

Then, we compute $\text{median}(G''')$. Since node $c$ is the median of both $G''$ and $G'''$, it is immediate to see that
\[
\text{median}(G''') - \text{median}(G'') = d(c,a') = \text{radius}(G'') = \text{diameter}(G'),
\]
and hence we can proceed as in the proof of Lemma~\ref{lem:SP-Diam}. 

In the directed case, we use the same weighted graphs $G''$ and $G'''$ as before, adding one parallel edge for each edge,
both with the same weight but with opposite directions. The rest of the algorithm is the same and its analysis is
analogous to the undirected case.
\end{proof}

Now we consider the evaluation of the betweenness centrality of nodes. In contrast to the previous reductions, in the following
one we shall create $n$ copies of the reduction graph leveraging Lemma~\ref{lem:replication}, and then perform some
computation in parallel.

\begin{lemma}\label{lem:SP-BC}
Shortest path is $O(1)$-round MPC reducible to betweenness centrality, with $O(n^{1-\epsilon})$ local memory
per machine for some constant $\epsilon \in (0,1)$, and $O(n(n+m))$ total memory.
\end{lemma}

\begin{proof}
Once again, we start with the case of undirected graphs. In the directed case we use the same weighted graph adding
one parallel edge for each edge, both with the same weight but with opposite directions, with an analogous analysis. 

Given an instance of shortest path, we construct the graph $G'$ of Figure~\ref{fig:SP-Diam} as in the reduction
from shortest path to diameter. Then, we modify the weights of the edges of $G'$ in such a way that exactly one
shortest path exists from any node to any other node, and that the length of the original shortest path in $G$
can be easily recovered. To this end, since by assumption the weights of the edges are integers, it is sufficient to
increase the weight of each edge of the starting graph $G$ by a real value chosen independently and uniformly
at random from $[1/n^3, 1/n^2]$. This reduction can be performed in $O(1)$ MPC rounds, it increases the number
of nodes and the number of edges by $O(n)$, and the maximum absolute weight by a factor of $O(n)$.

Now we create $n-1$ more copies of this graph, which by Lemma~\ref{lem:replication} can be done in $O(1)$ MPC
rounds, and compute the betweenness centrality of each node of $G$, in parallel on each copy of the graph. Since
there is a single shortest path from any node to any other node, the betweenness centrality of a node $u$ is the
total number of shortest paths in the graph that use $u$ as an intermediate node. Consider the (unique) shortest
path from $s$ to $t$, and let $A$ be the set of its nodes. Let $B = V \setminus A$ be the remaining nodes of $G$.
Then observe that, for any node $u \in A$, $\text{BC}(u) \geq 2n \cdot 2n$, and that for any node $u \in B$,
$\text{BC}(u) \leq \binom{n}{2}$. Hence, to compute the shortest path from $s$ to $t$ in $G$ it is sufficient
to consider only the nodes whose betweenness centrality is no less than $4n^2$, and return the sum of the floors of
the weights of all edges with both endpoints in this set of nodes. This can be easily done in $O(1)$ MPC rounds.
\end{proof}

An immediate consequence of these results is the following.

\begin{proposition}\label{pro:SP-equivalence}
Shortest path, SSSP, APSP, diameter, radius, median, and betweenness centrality are all equivalent under $O(1)$-round
MPC reductions, with $O(n^{1-\epsilon})$ local memory per machine for some constant $\epsilon \in (0,1)$, and
$O(n^2 (n+m))$ total memory.
\end{proposition}

\begin{proof}
The two reductions involving SSSP are obvious. The reduction from diameter (or radius) to APSP is also obvious,
since determining the maximum (or minimum) in a set of values can be easily done in $O(1)$ MPC rounds.
The theorem then follows from \cref{lem:APSP-SP,lem:SP-Diam,lem:SP-Radius,lem:SP-Median,lem:SP-BC}.
\end{proof}

It is interesting to observe that this equivalence class includes problems, such as SSSP and APSP, that in the
(sequential) RAM model have vastly different complexities, and that an analogous reduction from APSP to
diameter in the RAM model seems elusive~\cite{AbboudGW15}.


We can now summarize all the results of this section.

\begin{theorem}\label{thm:dconn-equivalence}
The following problems are all equivalent under $O(1)$-round MPC reductions, with $O(n^{1-\epsilon})$ local memory
per machine for some constant $\epsilon \in (0,1)$, and $O(n^{2}(n+m))$ total memory: $st$-reachability, strong connectivity,
directed cycle detection, unweighted shortest path, unweighted SSSP, unweighted APSP, unweighted diameter, unweighted radius,
unweighted median, and unweighted betweenness centrality.
\end{theorem}

\paragraph{Conditional hardness: problems hard for $\NL$ under $O(1)$-round MPC reductions.}
Finally, there are problems known to be hard for $\NL$ under $\AC^0$, and thus $\NC^1$ and $O(1)$-round MPC,
reductions, but not known to be in $\NL$. Some examples are perfect matching (even in bipartite graphs), network flow,
and circuit evaluation~\cite{ChandraSV84}. Since for these problems only one-way reductions from problems in $\NL$
are known, we don't know whether they are part of the equivalence class of directed graph connectivity.

\section{Open Problems}

The present work can be naturally extended in several directions. One obvious direction is to prove more conditional
lower bounds based on the conjectures of this paper, and to show more equivalences between problems.

Several results of this paper, from the connections between MPC computations and space complexity of Section~\ref{sec:MPC-Space}
to the reductions of Section~\ref{sec:reductions}, crucially require the availability in the system of a total amount of memory
super-linear in the size of the input. These results have no implications for the (perhaps more) interesting case of low
total memory---that is, linear or near-linear in the input size.\footnote{For instance, since the circuit simulation results
in~\cite{RoughgardenVW18,FreiW19} use super-linear total memory, the barrier for unconditional lower bounds they imply
falls for the linear total memory case, for which it is therefore still open the possibility of proving unconditional lower bounds.}
Hence, it would be interesting to establish equivalence classes and show implications that hold under more severe restrictions
on the total amount of available memory.

Finally, it is tempting to speculate that improved algorithms for any of the problems discussed in this paper could have
significant consequences in other models of computation, such as falsifying some widely-believed conjecture in complexity
theory. Identifying new consequences of their falsification would add further weight to the conjectures of this paper.

\section*{Acknowledgments}
This project has received funding from the European Research Council (ERC) under the European Union's Horizon 2020
research and innovation programme under grant agreement No 715672. M.~Scquizzato was also partially supported by
the University of Padova under grant BIRD197859/19.

\bibliographystyle{abbrv}
\bibliography{biblio}

\newpage

\appendix
\section*{APPENDIX}

\section{Space Complexity of Fundamental Problems}\label{sec:space}

Here we report what is known about the space complexity of several fundamental problems. Two good sources of problems
complete for $\L$ or $\NL$ are~\cite{CookM87,JonesLL76}.

\begin{description}
\item[Graph Connectivity] $\L$-complete: $\L$-hard~\cite[Theorem 3]{CookM87}, and in $\L$ by virtue of the remarkable algorithm
of Reingold~\cite{Reingold08}. Remains $\L$-complete for promise graphs that are a disjoint union of cycles~\cite[Theorem 3]{CookM87}.

\item[$st$-connectivity] $\L$-complete, by virtue of a non-obvious equivalence with graph connectivity under projection
reducibility shown by Chandra et al.~\cite{ChandraSV84}.

\item[$st$-reachability for directed graphs of out-degree one] This is the out-degree one version of $st$-reachability.
It is $\L$-complete~\cite{Jones75}.

\item[Order Between Vertices] Given a directed path, specified by giving for each node its successor in the path,
and two distinguished nodes $a$ and $b$, \emph{Order Between Vertices (ORD)}, sometimes also
called \emph{Path Ordering}, asks to determine whether $a$ precedes $b$. ORD is $\L$-complete~\cite{Etessami97}.

\item[Formula Evaluation] A \emph{formula} is a circuit where each gate has fan-out (out-degree) exactly one, where
the underlying algebraic structure is the Boolean algebra. Hence a formula is a circuit whose underlying graph is a tree.
It is easy to see that Boolean formula evaluation is in $\L$. A seminal paper by Buss shows that Boolean formula
evaluation belongs to $\NC^1$~\cite{Buss87}. However, for this result it is crucial that the Boolean formula is given
as a string (for instance its preorder notation), and not as a tree in pointer representation (e.g., by the list of all edges
plus gate types). For the latter representation, the problem is $\L$-complete~\cite{BeaudryM95}.

\item[Cycle Detection] $\L$-complete, even when the given graph contains at most one cycle~\cite{CookM87}.

\item[Planarity Testing] Is a given graph planar? Allender and Mahajan~\cite{AllenderM04} showed that this problem
is hard for $\L$ under projection reducibility (even for graphs of maximum degree 3), and that it lies in $\SL$.
Thus, by the result of Reingold~\cite{Reingold08}, planarity testing is $\L$-complete.

\item[Densest Subgraph] Given an undirected graph and a number $k$, is the density of a densest subgraph at least $k$?
Observe that cycle detection is a special case of this problem with $k=1$, and thus densest subgraph is $\L$-hard.

\item[$st$-reachability] This is $st$-connectivity in directed graphs, that is, the problem of detecting whether
there is a path from a distinguished node $s$ to a distinguished node $t$ in a directed graph. It is denoted STCON,
and also known as \emph{directed $st$-connectivity}, \emph{graph reachability}, \emph{PATH}, or \emph{graph
accessibility problem (GAP)}. It is the prototypical complete problem for $\NL$~\cite{Papadimitriou94,Sipser97,AroraB09}.
(This result was first proved by Jones~\cite{Jones75}, and is implicit in~\cite{Savitch70}, where STCON is called the
``threadable  maze'' problem.) It remains $\NL$-complete for the stronger case of first-order reductions~\cite{Immerman99}.
It is $\L$-hard~\cite{Savitch70,LewisP82}, but not known to be in $\L$. In Eulerian directed graphs (i.e., directed graphs
where each node has in-degree equal to its outdegree) it is in $\L$~\cite{ReingoldTV06}.

\item[Strong Connectivity] $\NL$-complete: equivalent to $st$-reachability under $\AC^0$ reductions~\cite{ChandraSV84}.

\item[Shortest Path] Given an undirected (unweighted) graph, two distinguished nodes $s$ and $t$, and an integer $k$,
the problem of determining if the length of a shortest path from $s$ to $t$ is $k$ is $\NL$-complete~\cite{BorodinCDRT89}.

\item[Directed Cycle Detection] Given a directed graph, does it contain a directed cycle? $\NL$-complete~\cite{Sipser97}.

\item[2SAT] $\NL$-complete~\cite[Theorem 16.3]{Papadimitriou94}.

\item[NFA/DFA Acceptance] $\NL$-complete~\cite{Sipser97}.

\item[Perfect Matching] $\NL$-hard, even in bipartite graphs, because of a $\AC^0$ reduction from $st$-reachability~\cite{ChandraSV84}.
(Also $\L$-hard, even on $k$-trees~\cite[Lemma 5.1]{DasDN13}.) It is a long-standing open question to determine whether
perfect matching is in $\NC$ (despite some recent substantial progress~\cite{SvenssonT17}).

\item[Bipartite Matching, Network Flow] Equivalent under $\AC^0$ reductions to bipartite perfect matching~\cite{ChandraSV84},
and thus $\NL$-hard.

\item[Circuit Evaluation] It is $\P$-complete under $\AC^0$ reductions~\cite{ChandraSV84}, and thus also $\NL$-hard and $\L$-hard.

\end{description}

\subsection{L-Completeness of the One Cycle vs.\ Two Cycles Problem}\label{sec:space_1vs2cycles}

In~\cite[Theorem 3]{CookM87} it is shown that graph connectivity when the given graph is known to be a disjoint union
of cycles is $\L$-hard. A careful inspection of the reductions used to establish this result reveals that
the problem remains hard even when the graph is known to be made up of either one or three cycles. By reducing
from a different problem, we now show that graph connectivity remains hard even when the graph is known to be
made up of either one or two cycles.\footnote{Graphs are allowed to have parallel edges, that is, cycles with two edges.}

\begin{proposition}
Graph connectivity for promise graphs that are either one cycle or two cycles is $\L$-complete.
\end{proposition}

\begin{proof}
Membership in $\L$ is guaranteed by the algorithm of Reingold~\cite{Reingold08}. To show $\L$-hardness, we shall exhibit
an $\NC^1$ reduction from order between vertices. Given an instance $(G,a,b)$ for order between vertices, we build a new
graph $G'$ as follows: (1) the two arcs pointing to $a$ and to $b$, denoted $(a',a)$ and $(b',b)$, respectively, are removed,
(2) the direction of each of the remaining $n-3$ arcs is discarded, and (3) edges $\{s,a\}$, $\{a',b'\}$, and $\{b,t\}$ are
added, where $s$ denotes the source and $t$ the sink of $G$, respectively. See Figure~\ref{fig:ORD-1vs2}.
This construction is an $\NC^1$ reduction. The resulting graph $G'$ consists of two cycles if $a$ precedes $b$ in $G$,
and of one single cycle otherwise.
\begin{figure}
    \begin{center}
       \includegraphics[trim=0cm 0.6cm 0cm 0.7cm, clip=true, width=\textwidth]{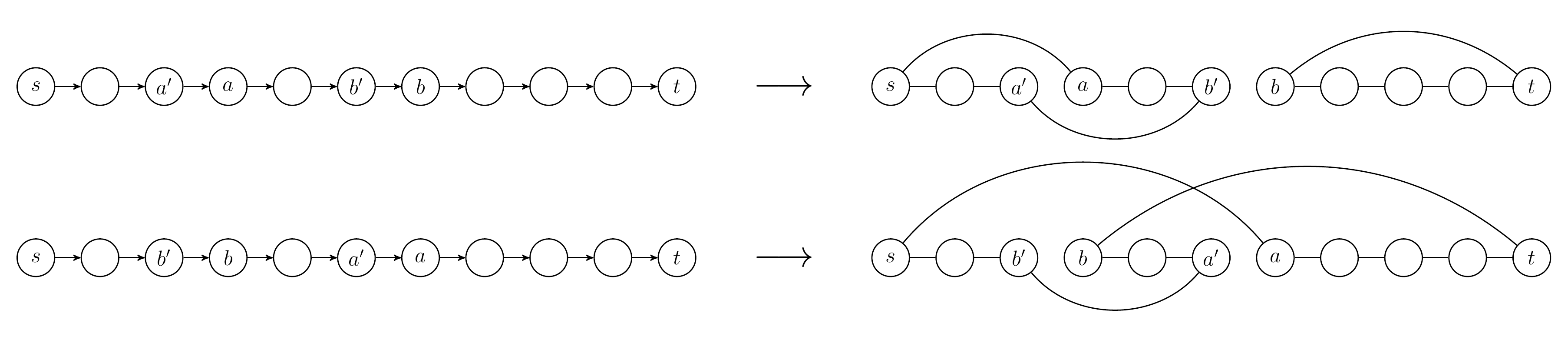}
       \caption{Reduction from order between vertices to one cycle vs.\ two cycles.}\label{fig:ORD-1vs2}
    \end{center}
\end{figure}
\end{proof}

\end{document}